\documentclass[submission,copyright,creativecommons]{eptcs}
\providecommand{\event}{GandALF 2018} 

\usepackage[final]{microtype}

\usepackage{graphicx}
\usepackage{tikz}
\usetikzlibrary{positioning,arrows}
\usepackage{paralist}

\usepackage{amsmath,amsfonts,amssymb,amsthm}
\usepackage{mathtools}

\DeclareMathOperator{\NP}{\mathbf{NP}}
\DeclareMathOperator{\Psp}{\mathbf{PSPACE}}
\DeclareMathOperator{\EXPSPACE}{\mathbf{EXPSPACE}}
\DeclareMathOperator{\NEXPTIME}{\textbf{NEXPTIME}}
\DeclareMathOperator{\NPSPACE}{\mathbf{NPSPACE}}

\DeclareMathAlphabet{\mathpzc}{OT1}{pzc}{m}{it}

\newcommand{\Beg}{\textit{beg}}
\newcommand{\End}{\textit{end}}

\newcommand{\tpl}[1]{(#1)}

\newcommand{\Nat}{\mathbb{N}}
\newcommand{\true}{\top}

\theoremstyle{plain}
\newtheorem{proposition}{Proposition}
\newtheorem{theorem}[proposition]{Theorem}

\theoremstyle{definition}
\newtheorem{definition}[proposition]{Definition}

\theoremstyle{remark}
\newtheorem{claim}[proposition]{Claim}

\newcommand{\details}[1]{{}}%
\newcommand{\dummy}{{\textit{dummy}}}
\newcommand{\Inst}{{\textit{L}}}
\newcommand{\halt}{{\textit{halt}}}
\newcommand{\init}{{\textit{init}}}
\newcommand{\main}{{\textit{main}}}

\newcommand{\inc}{{\mathsf{inc}}}
\newcommand{\dec}{{\mathsf{dec}}}
\newcommand{\zero}{{\mathsf{zero}}}
\newcommand{\instr}{{\textit{op}}}
\newcommand{\gainy}{{\textit{gainy}}}
\newcommand{\From}{{\textit{from}}}
\newcommand{\To}{{\textit{to}}}
\newcommand{\cont}{{\textit{sec}}}
\newcommand{\Tag}{{\textit{Tag}}}

\newcommand{\start}{\mathsf{s}}
\newcommand{\Ending}{\mathsf{e}}

\newcommand{\startTime}{\mathsf{s}}

\newcommand{\der}[1]{\ensuremath{\;\;{\mathop{{ %
            \longrightarrow}}\limits^{{#1}}}\!}\;\;} %

\newcommand{\derG}[1]{\ensuremath{\;\;{\mathop{{ %
            \longrightarrow}}\limits^{{#1}}}\!}_\gainy\;\;} %

\newcommand{\TA}{\text{\sffamily TA}}
\newcommand{\MTL}{\text{\sffamily MTL}}
\newcommand{\TPTL}{\text{\sffamily TPTL}}
\newcommand{\MITL}{\text{\sffamily MITL}}
\newcommand{\MITLR}{\text{\sffamily MITL}_{(0,\infty)}}
\newcommand{\LTL}{\text{\sffamily LTL}}

\newcommand{\RealP}{{\mathbb{R}_+}}

\newcommand{\val}{{\mathit{val}}}
\newcommand{\Main}{{\mathit{Main}}}
\newcommand{\EqTime}{{\mathit{EqTime}}}
\newcommand{\Past}{{\mathit{past}}}

\newcommand{\Res}{\textit{Res}}
\newcommand{\Au}{\ensuremath{\mathcal{A}}}
\newcommand{\TLang}{{\mathcal{L}_T}}
\newcommand{\Prop}{\mathcal{P}}
\newcommand{\StrictUntil}{\textsf{U}}

\newcommand{\Eventually}{\textsf{F}}
\newcommand{\Always}{\textsf{G}}
\newcommand{\Deriv}{{\mathit{Deriv}}}
\newcommand{\Intv}{{\mathit{Intv}}}
\newcommand{\IntvR}{{\mathit{Intv}_{(0,\infty)}}}

\title{Complexity of Timeline-Based Planning over Dense \\ Temporal Domains: Exploring the Middle Ground}

\author{Laura Bozzelli \qquad Adriano Peron
\institute{University of Napoli ``Federico II'', Napoli, Italy}
\email{lr.bozzelli@gmail.com \qquad adrperon@unina.it}
\and
Alberto Molinari \qquad Angelo Montanari
\institute{University of Udine, Udine, Italy}
\email{molinari.alberto@gmail.com \qquad angelo.montanari@uniud.it}
}

\def\titlerunning{Complexity of Timeline-Based Planning over Dense Domains}
\def\authorrunning{L.\ Bozzelli, A.\ Molinari, A.\ Montanari \& A.\ Peron}

\begin{document}

\maketitle

\begin{abstract}
In this paper, we address complexity issues for timeline-based planning over dense temporal domains.
The planning problem is modeled by means of a set of independent, but interacting, components, each one represented by a number of state variables, whose behavior over time (timelines) is governed by a set of temporal constraints (synchronization rules).
While the temporal domain is usually assumed to be discrete, here we consider the dense case.  
Dense timeline-based planning has been recently shown to be undecidable in the general case; decidability ($\NP$-completeness) can be recovered by restricting to purely existential synchronization rules (\emph{trigger-less rules}).
In this paper, we investigate the unexplored area of intermediate cases in between these two extremes.
We first show that \textbf{decidability} and \textbf{non-primitive recursive hardness} can be proved by admitting synchronization rules with a trigger, but forcing them to suitably check constraints only in the future with respect to the trigger (\emph{future simple rules}). More \lq\lq tractable\rq\rq{}
results can be obtained by additionally constraining the form of intervals in future 
simple rules: 
$\EXPSPACE$-completeness is guaranteed by avoiding singular intervals, 
 $\Psp$-completeness by admitting only intervals of the forms $[0,a]$ and $[b,+\infty[$.
\end{abstract}

\section{Introduction}
In this paper, we explore the middle ground of timeline-based planning over dense temporal domains.
%
\emph{Timeline-based planning} can be viewed as an alternative to the classical action-based approach to planning. Action-based planning aims at determining a sequence of actions that, given the initial state of the world 
and a goal, transforms, step by step, the state of the world until a state that satisfies the goal is reached. 
%
Timeline-based planning focuses on what has to happen in order to satisfy the goal instead of what an agent has to do. It models the planning domain as a set of independent, but interacting, components, each one consisting of a number of state variables. The evolution of the values of state variables over time is described by means of a set of timelines (sequences of tokens), and it is governed by a set of transition functions, one for each state variable, and a set of  synchronization rules, that constrain the temporal relations among state variables. Figure~\ref{fig:timelineEx} gives an account of these notions.

Timeline-based planning has been successfully exploited in a number of application domains, e.g.,~\cite{barreiro2012europa,CestaCFOP07,aspen2010,FrankJ03,JonssonMMRS00,Muscettola94}, but a systematic study of its expressiveness and complexity has been undertaken only very recently. The temporal domain is commonly assumed to be discrete, 
the dense case being dealt with by forcing an artificial discretization of the domain. 
In~\cite{GiganteMCO16}, Gigante et al.\ showed that timeline-based planning with bounded temporal relations and token durations, and no temporal horizon, is $\EXPSPACE$-complete and  expressive enough to capture action-based temporal planning. Later, in~\cite{GiganteMCO17}, they proved that $\EXPSPACE$-completeness still holds for timeline-based planning with unbounded interval relations, and that the problem becomes $\NEXPTIME$-complete if an upper bound to the temporal horizon is added. 

Timeline-based planning (TP for short) over a dense temporal domain has been studied in~\cite{kr18}.
Having recourse to a dense time domain is important for expressiveness: only in such a domain one can really express interval-based properties of planning domains, and can abstract from unnecessary (or even \lq\lq forced\rq\rq) details which are often artificially added due to the necessity of discretizing time.
The general TP problem  has been shown to be \textbf{undecidable} even when a single state variable is used. Decidability can be recovered by suitably constraining the logical structure of synchronization rules. In the general case, a synchronization rule allows a universal quantification over the tokens of a timeline (\emph{trigger}). By disallowing the universal quantification and retaining only rules in purely existential form (\emph{trigger-less rules}), the problem becomes $\NP$-complete~\cite{ictcs18}. These two bounds identify a large unexplored area of intermediate cases where to search for 
a balance between expressiveness and complexity. 
Investigating such cases is fundamental: as a matter of fact, trigger-less rules can essentially be used only to express initial conditions
and the goals of the problem, while trigger rules, much more powerful, are useful to specify invariants and response requirements. Thus one needs somehow a way of re-introducing the latter rules in order to recover their expressive power at least partially.

In this paper, we investigate 
the restrictions under which the universal quantification of triggers can be admitted though retaining decidability. 
When a token is \lq\lq selected\rq\rq{} by a trigger, the synchronization rule allows us to compare 
tokens of the timelines both preceeding (past) and following (future) the trigger token. 
The first restriction we consider consists in limiting the comparison to tokens in the future with respect to the trigger (\emph{future semantics of trigger rules}). 
The second restriction we consider imposes that, in a trigger rule, the name of a non-trigger token appears exactly once in the \emph{interval atoms} of the rule (\emph{simple trigger rules}). 
This syntactical restriction avoids comparisons of multiple token time-events with a non-trigger reference time-event.  
From the expressiveness viewpoint, even if we do not have a formal statement, we conjecture that future simple  trigger rules, together with arbitrary trigger-less rules allow for expressiveness strictly in between $\MTL$~\cite{AlurH93} and $\TPTL$~\cite{AlurH94}. 
Note that, by~\cite{kr18}, the TP problem with \emph{simple} trigger rules is already undecidable. 
In this paper, we show that it becomes decidable, although non-primitive recursive hard, under the \emph{future semantics} of the trigger rules. 
Better complexity results can be obtained by restricting also the type of intervals used in the simple trigger rules to compare tokens. In particular, we show that future TP with simple trigger rules without singular intervals%
\footnote{An interval is called \emph{singular} if it has the form $[a,a]$, for $a\in\Nat$.} %
is $\EXPSPACE$-complete. 
The problem is instead $\Psp$-complete if we consider only intervals of the forms $[0,a]$ and $[b,+\infty[$. 
The decidability status of the TP problem with arbitrary trigger rules under the future semantics remains open. 
However, we show that such a problem 
is at least non-primitive recursive even under the assumption that the intervals in the rules have the forms $[0,a]$ and $[b,+\infty[$.


\paragraph{Organization of the paper.} In Section~\ref{sec:preliminaries}, we recall the
TP framework. 
In Section~\ref{sec:DecisionProcedures}, we establish that future TP with simple trigger rules is decidable, and show membership in $\EXPSPACE$ (resp., $\Psp$) under the restriction to non-singular intervals (resp., intervals of the forms $[0,a]$ and $[b,+\infty[$). Matching lower bounds for the last two problems are given in Section~\ref{sec:pspace}.  In Section~\ref{sec:NPRHardness}, we prove  non-primitive recursive hardness of TP under the future semantics of trigger rules. 
Conclusions give an assessment of the work and outline future research themes.
All missing proofs and results can be found in \cite{techrepGand}.
\section{Preliminaries}\label{sec:preliminaries}

Let $\Nat$ be the set of natural numbers, $\RealP$ be the set of non-negative real numbers, and $\Intv$ be the set of intervals in $\RealP$ whose endpoints are in $\Nat\cup\{\infty\}$. Moreover, let us denote by $\Intv_{(0,\infty)}$ the set of intervals $I\in \Intv$ such that
  either $I$ is unbounded, or $I$
  is left-closed with left endpoint $0$. Such intervals $I$ can be replaced by expressions of the form $\sim n$ for some $n\in\Nat$
  and $\sim\in\{<,\leq,>,\geq\}$.
Let $w$ be a finite word over some alphabet. By $|w|$ we denote the length of $w$. For all  $0\leq i<|w|$,  $w(i)$ is
the $i$-th letter of $w$.

\subsection{The TP Problem}\label{sec:timelines}

In this section, we recall 
the TP framework as presented in \cite{MayerOU16,GiganteMCO16}.
In TP, domain knowledge is encoded by a set of state variables, whose behaviour over time is described by transition functions and synchronization rules.


\begin{definition}
  \label{def:statevar}
  A \emph{state variable} $x$ is a triple $x= (V_x,T_x,D_x)$, where $V_x$ is the \emph{finite domain} of the variable $x$, $T_x:V_x\to 2^{V_x}$ is the \emph{value transition function}, which maps
        each $v\in V_x$ to the (possibly empty) set of successor values, and $D_x:V_x\to \Intv$ is the \emph{constraint function} that maps each $v\in V_x$
        to an interval. 
\end{definition}

A \emph{token} for a variable $x$ is a pair $(v,d)$ consisting of a value $v\in V_x$ and a duration $d\in \RealP$
such that $d\in D_x(v)$. Intuitively, a token for $x$ represents an interval of time where the state variable $x$ takes value $v$.
The behavior of the state variable $x$ is specified by means of \emph{timelines} which are non-empty sequences of tokens
$\pi = (v_0,d_0)\ldots  (v_n,d_n)$  consistent with the value transition function $T_x$, that is, such that
$v_{i+1}\in T_x(v_i)$ for all $0\leq i<n$. The \emph{start time} $\start(\pi,i)$ and the \emph{end time} $\Ending(\pi,i)$ of the $i$-th token ($0\leq i\leq n$) of the timeline
$\pi$  are defined as follows: $\Ending(\pi,i)=\displaystyle{\sum_{h=0}^{i}} d_h$ and  $\start(\pi,i)=0$ if $i=0$, and $\start(\pi,i)=\displaystyle{\sum_{h=0}^{i-1}} d_h$ otherwise.
See Figure~\ref{fig:timelineEx} for an example.
\begin{figure}
    \centering
    \includegraphics{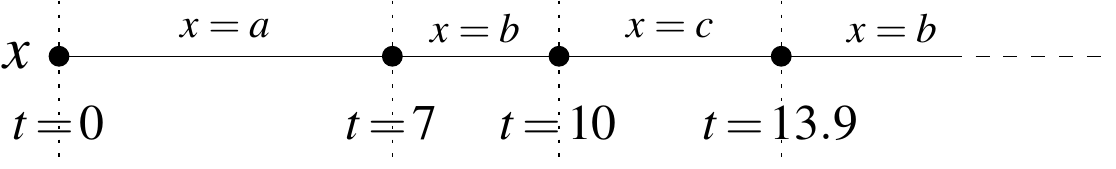}
    \caption{An example of timeline $(a,7)(b,3)(c,3.9)\cdots$ for the state variable $x= (V_x,T_x,D_x)$, where $V_x=\{a,b,c,\ldots\}$, $b\in T_x(a)$, $c\in T_x(b)$, $b\in T_x(c)$\dots and $D_x(a)=[5,8]$, $D_x(b)=[1,4]$, $D_x(c)=[2,\infty[$\dots}
    \label{fig:timelineEx}
\end{figure}




%


Given a finite set $SV$ of state variables, a \emph{multi-timeline} of $SV$ is a mapping $\Pi$ assigning to
each state variable $x\in SV$ a timeline for $x$.
Multi-timelines of $SV$ can be constrained by a set of \emph{synchronization
rules}, which relate tokens, possibly belonging to different timelines, through
temporal constraints on the start/end-times of tokens (time-point constraints) and on the difference
between start/end-times of tokens (interval constraints). The synchronization rules exploit
an alphabet $\Sigma$ of token names to refer to the tokens along a multi-timeline, and are based on the notions of
\emph{atom} and \emph{existential statement}.

%

\begin{definition}
  \label{def:timelines:atom}
  An \emph{atom} is either a clause of the form $o_1\leq^{e_1,e_2}_{I} o_2$
  (\emph{interval atom}), or of the forms $o_1\leq^{e_1}_{I} n$ or  $n\leq^{e_1}_{I}
  o_1$ (\emph{time-point atom}), where $o_1,o_2\in\Sigma$, $I\in\Intv$, $n\in\Nat$, and $e_1,e_2\in\{\start,\Ending\}$.
\end{definition}

An atom $\rho$ is evaluated with respect to a \emph{$\Sigma$-assignment $\lambda_\Pi$ for a given multi-timeline $\Pi$}
which is a mapping assigning to each token name $o\in \Sigma$ a pair $\lambda_\Pi(o)=(\pi,i)$ such that $\pi$ is a timeline of $\Pi$ and $0\leq i<|\pi|$ is a position along $\pi$ (intuitively,
$(\pi,i)$ represents the token of $\Pi$ referenced by the name $o$).
An interval atom $o_1\leq^{e_1,e_2}_{I} o_2$  \emph{is satisfied by  $\lambda_\Pi$} if $e_2(\lambda_\Pi(o_2))-e_1(\lambda_\Pi(o_1))\in I$.
A point atom $o\leq^{e}_{I} n$  (resp., $n\leq^{e}_{I}o$)   \emph{is satisfied by  $\lambda_\Pi$} if $n-e(\lambda_\Pi(o))\in I$ (resp., $e(\lambda_\Pi(o))-n\in I$).


\begin{definition}
 An \emph{existential statement} $\mathcal{E}$ for a finite set $SV$ of state variables is a statement of the form:
\[
\mathcal{E}:=  \exists o_1[x_1=v_1]\cdots \exists o_n[x_n=v_n].\mathcal{C}
\]
  where $\mathcal{C}$ 
  is a conjunction of atoms,
  $o_i\in\Sigma$, $x_i\in SV$, and $v_i\in V_{x_i}$ for each
  $i=1,\ldots,n$. The elements $o_i[x_i=v_i]$ are called
  \emph{quantifiers}. A token name used in $\mathcal{C}$, but not occurring in any
  quantifier, is said to be \emph{free}. Given a $\Sigma$-assignment $\lambda_\Pi$ for a multi-timeline $\Pi$ of $SV$,
  we say that \emph{$\lambda_\Pi$ is consistent with the existential statement $\mathcal{E}$} if for each quantified token name $o_i$,
   $\lambda_\Pi(o_i)=(\pi,h)$ where $\pi=\Pi(x_i)$ and the $h$-th token of $\pi$ has value $v_i$. A multi-timeline $\Pi$ of $SV$ \emph{satisfies} $\mathcal{E}$
   if there exists a $\Sigma$-assignment $\lambda_\Pi$ for $\Pi$ consistent with $\mathcal{E}$ such that each atom in $\mathcal{C}$ is satisfied by
   $\lambda_\Pi$.
\end{definition}

\begin{definition}
  A \emph{synchronization rule} $\mathcal{R}$ for a finite set $SV$ of state variables is a rule of one of the forms
  \[
  o_0[x_0=v_0] \to \mathcal{E}_1\lor \mathcal{E}_2\lor \ldots \lor \mathcal{E}_k, \quad
          \true \to \mathcal{E}_1\lor \mathcal{E}_2\lor \ldots \lor \mathcal{E}_k,
  \]
  where $o_0\in\Sigma$, $x_0\in SV$, $v_0\in V_{x_0}$, and $\mathcal{E}_1, \ldots, \mathcal{E}_k$
  are \emph{existential statements}. 
  In rules of the first
  form (\emph{trigger rules}), the quantifier $o_0[x_0=v_0]$ is called \emph{trigger}, and we require that only $o_0$ may appear free in $\mathcal{E}_i$ (for $i=1,\ldots,n$). In rules of the second form (\emph{trigger-less rules}), we require
  that no token name appears free.
  \newline
  A trigger rule $\mathcal{R}$ is \emph{simple} if for each existential statement $\mathcal{E}$ of $\mathcal{R}$ and each token name $o$ distinct from the trigger, there is at most one \emph{interval atom}
  of $\mathcal{E}$ where $o$ occurs.
\end{definition}

Intuitively, a  trigger $o_0[x_0=v_0]$ acts as a universal quantifier, which
states that \emph{for all} the tokens of the timeline for
the state variable $x_0$, where the variable $x_0$ takes the
value $v_0$, at least one of the existential statements $\mathcal{E}_i$ must be true. 
Trigger-less rules simply assert the
satisfaction of some existential statement. The intuitive meaning of the \emph{simple} trigger rules is that they disallow simultaneous comparisons of multiple time-events
 (start/end times of tokens) with a non-trigger reference time-event. The semantics of synchronization rules is formally defined as follows.

\begin{definition}\label{def:semanticsRules}
Let $\Pi$ be a multi-timeline of a set $SV$ of state variables.
Given a \emph{trigger-less rule} $\mathcal{R}$ of $SV$, \emph{$\Pi$ satisfies $\mathcal{R}$} if $\Pi$ satisfies some existential statement of $\mathcal{R}$.
 Given a \emph{trigger rule} $\mathcal{R}$ of $SV$ with trigger $o_0[x_0=v_0]$, \emph{$\Pi$ satisfies   $\mathcal{R}$} if for every position $i$ of the
 timeline $\Pi(x_0)$ for $x_0$ such that $\Pi(x_0)=(v_0,d)$, there is an existential statement $\mathcal{E}$ of $\mathcal{R}$  and a $\Sigma$-assignment
 $\lambda_\Pi$ for $\Pi$ which is consistent with $\mathcal{E}$ such that $\lambda_\Pi(o_0)= (\Pi(x_0),i)$ and $\lambda_\Pi$ satisfies all the atoms of $\mathcal{E}$.

\end{definition}

In the paper, we focus on a stronger notion of satisfaction of trigger rules, called \emph{ satisfaction under the future semantics}. It requires that all the non-trigger selected tokens
do not start \emph{strictly before} the start-time of the trigger token.

\begin{definition}
  A multi-timeline $\Pi$ of $SV$  \emph{satisfies under the future semantics} a trigger rule $\mathcal{R}= o_0[x_0=v_0] \to \mathcal{E}_1\vee \mathcal{E}_2\vee
  \ldots \vee \mathcal{E}_k$   if $\Pi$ satisfies the trigger rule obtained from
  $\mathcal{R}$ by replacing each existential statement $\mathcal{E}_i=\exists o_1[x_1=v_1]\cdots \exists o_n[x_n=v_n].\mathcal{C}$
  with  $\exists o_1[x_1=v_1]\cdots \exists o_n[x_n=v_n].\mathcal{C}\wedge  \bigwedge_{i=1}^{n} o_0\leq^{\start,\start}_{[0,+\infty[} o_i$.
\end{definition}


A TP domain  $P=(SV,R)$ is specified by a finite set $SV$ of state variables and
a finite set $R$ of synchronization rules modeling their admissible behaviors.
 A \emph{plan of $P$} is a  multi-timeline of $SV$ satisfying all the rules in $R$. A \emph{ future plan of $P$} is defined in a similar way, but we require that the fulfillment of the trigger rules is under the future semantics.
We are interested in the following decision problems:
\begin{inparaenum}[(i)]
  \item \emph{TP problem:} given a TP domain $P=(SV,R)$, is there a plan for $P$?
  \item \emph{Future TP problem:} similar to the previous one, but we  require the existence of a future plan.
\end{inparaenum}

Table~\ref{tab:complex} summarizes all the decidability and complexity results described in the following.
We consider mixes of restrictions of the TP problem involving trigger rules with future semantics, simple trigger rules, and  intervals in atoms of trigger rules which are non-singular or in $\Intv_{(0,\infty)}$.

\begin{table}[t]
    \centering
    \resizebox{\linewidth}{!}{
    \begin{tabular}{r|c|c}
    	& TP problem & Future TP problem \\ 
    	\hline 
    	Unrestricted & Undecidable & (Decidable?) Non-primitive recursive-hard \\ 
    	\hline 
    	Simple trigger rules & Undecidable & Decidable (non-primitive recursive) \\ 
    	\hline 
    	Simple trigger rules, non-singular intervals & ? & $\EXPSPACE$-complete \\ 
    	\hline 
    	Simple trigger rules, intervals in $\Intv_{(0,\infty)}$ & ? & $\Psp$-complete \\ 
    	\hline 
    	Trigger-less rules only & $\NP$-complete & // \\ 
    \end{tabular} }
    \caption{Decidability and complexity of restrictions of the TP problem.}
    \label{tab:complex}
\end{table}


\section{Solving the future TP problem with simple trigger rules}\label{sec:DecisionProcedures}

Recently, we have shown that the TP problem is undecidable even if the trigger rules are assumed to be simple~\cite{kr18}.
In this section, we show that decidability can be recovered assuming that the trigger rules are \emph{simple} and \emph{interpreted under the future
semantics}. Moreover, we establish that under the additional assumption that intervals in trigger rules are non-singular (resp., are in $\Intv_{(0,\infty)}$), the problem is
in $\EXPSPACE$ (resp., $\Psp$). The decidability status of future TP with arbitrary trigger rules remains open.
In Section~\ref{sec:NPRHardness}, we prove that the latter problem is at least non-primitive recursive even if intervals in rules and in the constraint functions of the state variables
are assumed to be in $\Intv_{(0,\infty)}$.

The rest of this section is organized as follows: in Subsection~\ref{sec:TimedAutomata}, we recall the framework
of Timed Automata (\TA)~\cite{ALUR1994183}  and Metric Temporal logic (\MTL)~\cite{Koymans90}. In Subsection~\ref{sec:Reduction}, we reduce the future TP problem
with simple trigger rules to the existential model-checking problem for \TA\ against \MTL\ over \emph{finite timed words}. The latter problem is known to be decidable~\cite{OuaknineW07}.

\subsection{Timed automata and the logic \MTL}\label{sec:TimedAutomata}

Let us recall the notion of timed automaton (\TA)~\cite{ALUR1994183} and the logic \MTL~\cite{Koymans90}.
Let $\Sigma$ be a finite alphabet. A  \emph{timed word} $w$ over  $\Sigma$ is
a \emph{finite}  word $w=(a_0,\tau_0)\cdots (a_n,\tau_n)$ over $\Sigma\times \RealP$ 
($\tau_i$ is the time at which $a_i$ occurs) such that  $\tau_{i}\leq \tau_{i+1}$ for all $0\leq i<n$ (monotonicity).
The timed word $w$ is also denoted by $(\sigma,\tau)$, where $\sigma$ is the finite untimed  word $a_0 \cdots a_n$
and $\tau$ is the sequence of timestamps $\tau_0 \cdots \tau_n$.
A \emph{timed language}  over $\Sigma$ is a set of timed words over $\Sigma$.

\paragraph{Timed Automata (\TA).} Let $C$ be a finite set of clocks. A clock valuation $\val:C\to \RealP$ for $C$ is a mapping
 assigning a non-negative real value to each clock in $C$.
 For $t\in\RealP$ and a reset set $\Res\subseteq C$, $(\val+ t)$ and $\val[\Res]$ denote the valuations over $C$ defined as follows: for all $c\in C$,
 $(\val +t)(c) = \val(c)+t$, and $\val[\Res](c)=0$ if $c\in \Res$ and $\val[\Res](c)=\val(c)$ otherwise.
 A \emph{clock constraint} over $C$ is a conjunction of atomic formulas of the form
$c \in I$ or $c-c'\in I$, where $c,c'\in C$ and
$I\in\Intv$.
For a clock valuation $\val$ and a clock constraint $\theta$, $\val$ satisfies $\theta$, written
$\val\models \theta$, if, for each conjunct $c\in I$ (resp., $c-c'\in I$) of $\theta$, $\val(c)\in I$ (resp., $\val(c)-\val(c')\in I$).
We denote by $\Phi(C)$ the set of clock constraints over $C$.

\begin{definition}
 A  \TA\ over  $\Sigma$ is a tuple
$\Au=\tpl{\Sigma, Q,q_0,C,\Delta,F}$, where $Q$ is a finite
set of (control) states, $q_0\in Q$ is the initial
state, $C$ is the finite set of clocks,
$F\subseteq Q$ is the set of accepting states, and $\Delta \subseteq Q\times \Sigma \times \Phi(C) \times 2^{C} \times Q $ is the transition relation.
The \emph{maximal constant of $\Au$} is the greatest integer occurring as endpoint of some interval in the clock constraints of $\Au$.
\end{definition}

Intuitively, in a \TA\  $\Au$, while transitions are instantaneous, time can elapse in a control
state. The clocks  progress at the same speed  and can
be reset independently of each other when a transition is executed, in such a way that each clock
keeps track of the time elapsed since the last reset. Moreover, clock constraints
are used as guards of transitions to restrict the behavior of the
automaton.

Formally, a configuration of $\Au$ is a pair $(q,\val)$, where $q\in Q$ and $\val$ is a clock valuation for $C$.
A run $r$ of $\Au$ on a timed word $w\!=\!(a_0,\tau_0)\cdots (a_n,\tau_n)$ over $\Sigma$
is a sequence  of configurations
 $r=(q_0,\val_0)\cdots (q_{n+1},\val_{n+1})$ starting at the initial configuration $(q_0,\val_0)$,
with $\val_0(c)\!=\!0$ for all $c\!\in\! C$ and 
\begin{compactitem}
\item for all $0\leq i\leq n$ we have (we let $\tau_{-1}=0$): 
  $(q_{i},a_i,\theta,\Res,q_{i+1})\in\Delta$ for some $\theta\in\Phi(C)$ and reset set $\Res$, $(\val_{i} +\tau_i-\tau_{i-1})\models \theta$ and $\val_{i+1}= (\val_{i} +\tau_i-\tau_{i-1})[\Res]$.
\end{compactitem}
The run $r$ is \emph{accepting} if   $q_{n+1}\in F$.
The \emph{timed language} $\TLang(\Au)$ of $\Au$ is the set of  timed words $w$ over $\Sigma$
such that there is an accepting run of $\Au$ over $w$.

\paragraph{The logic \MTL.} We now recall the framework of Metric Temporal Logic (\MTL)~\cite{Koymans90},  a well-known  timed linear-time temporal logic which extends standard \LTL\ with time
constraints on until modalities.

For a finite set $\Prop$ of atomic propositions, the set of \MTL\ formulas $\varphi$ over $\Prop$ is defined as follows:
\[
\varphi::= \top \mid
p \mid
\varphi \vee \varphi \mid
\neg \varphi      \mid
\varphi \StrictUntil_I\varphi
\]
where $p\in \Prop$, $I\in\Intv$, and
$\StrictUntil_I$  is the standard \emph{strict timed  until} \MTL\ modality. \MTL\ formulas over $\Prop$ are interpreted over  timed words over $2^{\Prop}$.
Given an \MTL\ formula $\varphi$, a  timed word $w=(\sigma,\tau)$ over $2^{\Prop}$, and a position $0\leq  i< |w|$, the satisfaction relation
$(w,i)\models\varphi$, meaning that $\varphi$ holds at position $i$ of $w$, is  defined as follows (we omit the clauses for atomic propositions and Boolean connectives):
\begin{compactitem}
\item $(w,i)\models \varphi_1 \StrictUntil_I\varphi_2
              \Leftrightarrow   \text{there is   }  j>  i \text{ such that } (w,j)\models \varphi_2, \tau_j-\tau_i\in I, \text{ and }
              (w,k)\models \varphi_1 \text{ for all } i<k<j$.
\end{compactitem}
A \emph{model of $\varphi$} is a  timed word $w$ over $2^{\Prop}$ such that $(w,0)\models \varphi$. The \emph{timed language} $\TLang(\varphi)$ of $\varphi$ is the set of  models of $\varphi$.
The \emph{existential model-checking problem for \TA\ against \MTL} is the problem of checking for a \TA\ $\Au$ over $2^{\Prop}$ and an \MTL\ formula $\varphi$ over $\Prop$ whether
$\TLang(\Au)\cap \TLang(\varphi)\neq \emptyset$.

In \MTL, we use standard shortcuts: $\Eventually_I \varphi$ stands for $\varphi \vee (\true 
\StrictUntil_I \varphi)$ (\emph{timed eventually}), and $\Always_I \varphi$ stands for $\neg \Eventually_I  \neg\varphi$ (\emph{timed always}).
 We also consider two fragments of \MTL, namely, \MITL\ (Metric Interval Temporal Logic)
and $\MITLR$~\cite{Alur:1996}: \MITL\ is obtained by allowing only non-singular intervals in $\Intv$, while $\MITLR$ is the fragment of \MITL\ obtained by allowing only intervals
in $\IntvR$. The \emph{maximal constant} of an \MTL\ formula $\varphi$ is the greatest integer occurring as an endpoint of some interval of (the occurrences of) the $\StrictUntil_I$ modality in $\varphi$.

\subsection{Reduction to existential model checking of \TA\ against \MTL}\label{sec:Reduction}

In this section, we solve the future TP problem with simple trigger rules by an exponential-time reduction to the existential model-checking problem for \TA\ against \MTL. 

In the following, we fix an instance $P=(SV,R)$ of the problem such that the trigger rules in $R$ are simple. The \emph{maximal constant} of $P$, denoted by $K_P$, is the greatest integer occurring in the atoms of $R$ and in the constraint  functions of the variables in $SV$.

The proposed reduction consists of three steps:
\begin{inparaenum}[(i)]
  \item first, we define an encoding of the multi-timelines of $SV$ by means of timed words over $2^{\Prop}$ for a suitable finite set $\Prop$ of atomic propositions,
  and show how to construct a \TA\ $\Au_{SV}$ over $2^{\Prop}$ accepting such encodings;
  \item next, we build an \MTL\ formula $\varphi_{\forall}$ over $\Prop$ such that for each multi-timeline $\Pi$ of $SV$ and encoding $w_\Pi$ of $\Pi$, $w_\Pi$ is a model of $\varphi_\forall$
  if and only if $\Pi$ satisfies all the trigger rules in $R$ under the future semantics;
  \item finally, we construct a \TA\ $\Au_{\exists}$ over $2^{\Prop}$ such that for each multi-timeline $\Pi$ of $SV$ and encoding $w_\Pi$ of $\Pi$, $w_\Pi$ is accepted by $\Au_{\exists}$
  if and only if $\Pi$ satisfies all the trigger-less rules in $R$.
\end{inparaenum}
Hence, there is a future plan of $(SV,R)$  if and only if  $\TLang(\Au_{SV})\cap \TLang(\Au_\exists )\cap \TLang(\varphi_\forall)\neq \emptyset$. 

For each $x\in SV$, let $x=\tpl{V_x,T_x,D_x}$.
Given an interval $I\in\Intv$ and a natural number $n\in \Nat$, $n+I$ (resp., $n-I$) denotes   the set of non-negative real numbers
$\tau\in\RealP$ such that $\tau-n\in I$ (resp., $n-\tau \in I$). Note that $n+I$ (resp., $n-I$) is a (possibly empty) interval in $\Intv$ whose endpoints can be trivially calculated.
For an atom $\rho$ in $R$ involving a time constant (\emph{time-point atom}), let $I(\rho)$ be the interval in $\Intv$ defined as follows:
\begin{compactitem}
  \item if $\rho$ is of the form  $o\leq^{e}_{I} n$ (resp., $n\leq^{e}_{I} o$), then $I(\rho):= n-I$ (resp., $I(\rho)= n+I$).
\end{compactitem}
We define $\Intv_R:=\{J\in\Intv \mid J=I(\rho) \text{ for some time-point atom $\rho$ occurring in a trigger rule of } R\}$.

\paragraph{Encodings of multi-timelines of $SV$.} We assume that for distinct state variables $x$ and $x'$, the sets $V_x$ and  $V_{x'}$
are disjunct. We exploit the following set $\Prop$ of propositions to encode multi-timelines of $SV$: 
\[
\Prop := \displaystyle{\bigcup_{x\in SV}}\Main_x \cup \Deriv ,
\]
\[
\Main_x := (\{\Beg_x \}\cup V_x) \times V_x   \cup   V_x \times \{\End_x\}, \quad \quad \Deriv := \Intv_R \cup \{p_>\} \cup \bigcup_{x\in SV}\bigcup_{v\in V_x}\{\Past_v^{\start},\Past_v^{\Ending}\}.
\]
Intuitively, we use the propositions in $\Main_x$ to encode a token along a timeline for $x$. The propositions in $\Deriv$, as explained below, represent
enrichments of the encoding, used for translating simple trigger rules in \MTL\ formulas under the future semantics.
 The tags $\Beg_x$ and $\End_x$ in $\Main_x$ are used to mark the start and the end of a timeline  for $x$. In particular, a token $tk$ with value $v$ along a timeline for
 $x$ is encoded by two events: the \emph{start-event} (occurring at the start time of $tk$) and
 the \emph{end-event} (occurring at the end time of $tk$). The start-event of $tk$ is specified by a main proposition of the form
 $(v_p,v)$, where either $v_p=\Beg_x$ ($tk$ is the first token of the timeline) or $v_p$ is the value of the $x$-token
preceding $tk$. The end-event of $tk$ is instead specified by a main proposition of the form
 $(v,v_s)$, where either $v_s=\End_x$ ($tk$ is the last token of the timeline) or $v_s$ is the value of the $x$-token
following $tk$. Now, we explain the meaning of the propositions in $\Deriv$.
Elements in  $\Intv_R$ reflect the semantics of
the time-point atoms in the trigger rules of $R$: for each $I\in \Intv_R$, $I$ holds at the current position if the current timestamp $\tau$ satisfies
$\tau\in I$.  The tag $p_>$ keeps track of whether the current timestamp is strictly greater than the previous one.
Finally,  the propositions in $\bigcup_{x\in SV}\bigcup_{v\in V_x}\{\Past_v^{\start},\Past_v^{\Ending}\}$ keep track of past token events occurring at timestamps \emph{coinciding}
with the current timestamp. We first define the encoding of timelines for $x\in SV$.

A \emph{code for a timeline for $x$} is
 a timed word $w$ over $2^{\Main_x \cup \Deriv}$ of the form
 \[
 w =(\{(\Beg_x,v_0)\}\cup S_0,\tau_0),(\{(v_0,v_1)\}\cup S_1,\tau_1)\ldots (\{(v_n,\End_x)\}\cup S_{n+1},\tau_{n+1})
 \]
 where, for all $0\leq i\leq n+1$, $S_i\subseteq \Deriv$, and 
 \begin{inparaenum}[(i)]
  \item $v_{i+1}\in T_x(v_i)$ if $i<n$;
 \item $\tau_0=0$ and $\tau_{i+1}-\tau_i \in D_x(v_i)$ if $i\leq n$;
  \item   $S_i\cap \Intv_R$ is the set of intervals $I\in\Intv_R$
  such that $\tau_i\in I$, and $p_>\in S_i$ iff either $i=0$ or $\tau_i>\tau_{i-1}$;
  \item for all $v\in V_x$, $\Past^{\start}_v\in S_i$ (resp., $\Past^{\Ending}_v\in S_i$) iff there is $0\leq h<i$ such that $\tau_h=\tau_i$ and $v= v_h$ (resp., $\tau_h=\tau_i$, $v=v_{h-1}$ and $h>0$).
\end{inparaenum}
Note that the length of $w$ is at least $2$. The timed word $w$ encodes the timeline for $x$ of length $n+1$
given by $\pi=(v_0,\tau_1) (v_1,\tau_2-\tau_1)\cdots (v_n,\tau_{n+1}-\tau_n)$. Note that in the encoding, $\tau_i$ and $\tau_{i+1}$ represent the start time and the end time
of the $i$-th token of the timeline ($0\leq i\leq n$).

Next, we define the encoding of a multi-timeline for $SV$.  For a set $P\subseteq \Prop$ and $x\in SV$, let $P[x]:= P\setminus \bigcup_{y\in SV\setminus \{x\}}\Main_y$.
A \emph{code for a multi-timeline for $SV$} is
 a timed word $w$ over $2^{\Prop}$ of the form
$
 w =(P_0,\tau_0)\cdots (P_n,\tau_n)
$
 such that  the following conditions hold:
 \begin{inparaenum}[(i)]
  \item  for all $x\in SV$, the timed word obtained from $(P_0[x],\tau_0)\cdots (P_n[x],\tau_n)$ by removing
  the pairs $(P_i[x],\tau_i)$ such that $P_i[x]\cap \Main_x=\emptyset$ is a code of a timeline for $x$;
    \item $P_0[x]\cap \Main_x\neq \emptyset$ for all $x\in SV$ (initialization).
\end{inparaenum}

One can easily construct a \TA\ $\Au_{SV}$ over $2^{\Prop}$ accepting the encodings of the multi-timelines of $SV$.
In particular, the \TA\ $\Au_{SV}$ uses a clock $c_x$ for each state variable  $x$ for checking the time constraints on the duration of the tokens for $x$. Two additional clocks $c_>$
and $c_{glob}$ are exploited for capturing the meaning of the proposition $p_>$ and of the propositions in $\Intv_R$ (in particular, $c_{glob}$ is a clock which measures the current time and is never reset).
Hence, we obtain the following result (for details, see~\cite{techrepGand}).

\newcounter{prop-AtutomataForMultiTimeline}
\setcounter{prop-AtutomataForMultiTimeline}{\value{proposition}}

 \begin{proposition}\label{prop:AtutomataForMultiTimeline} One can construct a \TA\ $\Au_{SV}$ over $2^{\Prop}$, with $2^{O(\sum_{x\in SV}|V_x|)}$ states, $|SV|+2$ clocks, and
 maximal constant $O(K_P)$, such that $\TLang(\Au_{SV})$ is the set of codes for the multi-timelines of $SV$.
 \end{proposition}

 \paragraph{Encodings of simple trigger rules by \MTL\ formulas.} We now construct an \MTL\ formula $\varphi_{\forall}$ over $\Prop$ capturing the trigger rules in $R$, which, by hypothesis, are simple,
  under the future semantics.

 \begin{proposition}\label{prop:MTLTriggerRules} If the trigger rules in $R$ are simple, then one can construct in linear-time an \MTL\ formula $\varphi_{\forall}$, with maximal constant $O(K_P)$,
  such that for each multi-timeline $\Pi$ of $SV$ and encoding $w_\Pi$ of $\Pi$, $w_\Pi$ is a model of $\varphi_\forall$
  iff $\Pi$ satisfies all the trigger rules in $R$ under the future semantics. Moreover,
  $\varphi_\forall$ is an \MITL\ formula (resp., $\MITLR$ formula) if the intervals in the trigger rules are non-singular (resp., belong to $\IntvR$). Finally,
  $\varphi_\forall$ has $O(\sum_{x\in SV}|V_x| +N_a)$ distinct subformulas, where $N_a$ is the overall number of atoms in the
trigger rules in $R$.
 \end{proposition}
\begin{proof} We first introduce some auxiliary propositional (Boolean) formulas over $\Prop$. Let $x\in SV$ and $v\in V_x$. We denote by
$\psi(\start,v)$ and $\psi(\Ending,v)$ the two propositional formulas over $\Main_x$ defined as follows:
\[
\psi(\start,v):= (\Beg_x,v)\vee \displaystyle{\bigvee_{u\in V_x}}(u,v)\quad\quad \psi(\Ending,v):= (v,\End_x)\vee \displaystyle{\bigvee_{u\in V_x}}(v,u)
\]
Intuitively, $\psi(\start,v)$ (resp., $\psi(\Ending,v)$) states that a start-event (resp., end-event) for a token for $x$ with value $v$ occurs at the current time.
We also exploit the formula $\psi_{\neg x}:= \neg \bigvee_{m\in \Main_x} m$ asserting that no event for a token for $x$ occurs at the current time.
Additionally, for an \MTL\ formula $\theta$, we exploit the \MTL\ formula $\EqTime(\theta): = \theta \vee [\neg p_> \StrictUntil_{\geq 0}(\neg p_> \wedge \theta)]$ which is satisfied
by a code of a multi-timeline of $SV$ at the current time, if $\theta$ eventually holds at a position whose timestamp coincides with the current timestamp.

The \MTL\ formula $\varphi_{\forall}$ has a conjunct   
 $\varphi_{\mathcal{R}}$ for each trigger rule $\mathcal{R}$. 
 Let $\mathcal{R}$ be a trigger rule of the form
 $o_t[x_t =v_t] \to \mathcal{E}_1\vee \mathcal{E}_2\vee \ldots \vee \mathcal{E}_k$. 
 Then, the \MTL\ formula  $\varphi_{\mathcal{R}}$ is given by 
 \[
\varphi_{\mathcal{R}}:= \Always_{\geq 0} \big(\psi(\start,v_t) \rightarrow \displaystyle{\bigvee_{i=1}^{k}}\Phi_{\mathcal{E}_i}\big)
 \]
where the \MTL\ formula   $\Phi_{\mathcal{E}_i}$, with $1\leq i\leq k$, ensures the fulfillment of the existential statement $\mathcal{E}_i$
of the trigger rule $\mathcal{R}$ under the future semantics. Let $\mathcal{E}\in \{\mathcal{E}_1,\ldots,\mathcal{E}_k\}$, $O$ be the set of token names existentially quantified
by $\mathcal{E}$, $\mathbf{A}$ be  the set of \emph{interval} atoms in $\mathcal{E}$, and, for each $o\in O$, $v(o)$  be the value of the token referenced by $o$ in the associated quantifier. 
In the construction of $\Phi_{\mathcal{E}}$, we crucially exploit the assumption that  $\mathcal{R}$ is simple: 
for each token name $o\in O$, there is at most one atom in $\mathbf{A}$ where $o$ occurs.

For each token name $o\in \{o_t\}\cup O$, 
we denote by $\Intv_o^{\start}$ (resp., $\Intv_o^{\Ending}$) the set of intervals $J\in\Intv$ such that $J=I(\rho)$ for some time-point atom $\rho$ occurring in  $\mathcal{E}$, which imposes a time constraint on the start time (resp., end time) of the token referenced by $o$. Note that $\Intv_o^{\start},\Intv_o^{\Ending}\subseteq \Prop$, and we exploit the propositional formulas $\xi^{\start}_o = \bigwedge_{I\in \Intv^{\start}_o}I$ and $\xi^{\Ending}_o = \bigwedge_{I\in \Intv^{\Ending}_o}I$  to ensure the fulfillment of the time constraints imposed by the
time-point atoms associated with the token $o$.  
The \MTL\ formula $\Phi_{\mathcal{E}}$ is given by:
\[
\Phi_{\mathcal{E}}:=\xi^{\start}_{o_t} \wedge [\psi_{\neg x_t}\StrictUntil_{\geq 0}(\psi(\Ending,v_t)\wedge \xi^{\Ending}_{o_t})]\wedge \displaystyle{\bigwedge_{\rho\in \mathbf{A}}} \chi_\rho,
\]
where, for each atom $\rho\in \mathbf{A}$, the formula $\chi_\rho$ captures the future semantics of $\rho$. 

The construction of $\chi_\rho$ depends on the form of $\rho$. We distinguishes four cases.
\begin{compactitem}
   \item $\rho = o \leq_I^{e_1,e_2} o_t$  and $o\neq o_t$. We assume $0\in I$ (the other case being simpler). First, assume that $e_2=\start$. Under the future semantics,
  $\rho$ holds iff the start time of the trigger token $o_t$ coincides with the $e_1$-time of token $o$. Hence, in this case ($e_2=\start$), $\chi_\rho$ is given by:
  \[
  \chi_\rho := \xi_o^{e_1}\wedge  \bigl(\Past_{v(o)}^{e_1}\vee \EqTime(\psi(e_1,v(0)))\bigr).
  \]
  If instead $e_2 = \Ending$, then $\chi_\rho$ is defined as follows:
 \begin{multline*}
  \chi_\rho :=  [\psi_{\neg x_t}\StrictUntil_{\geq 0}\{\xi_o^{e_1}\wedge \psi(e_1,v(0))\wedge \psi_{\neg x_t} \wedge (\psi_{\neg x_t}\StrictUntil_I \psi(\Ending,v_t))\}]   \vee     [(\psi(e_1,v(0))\vee \Past_{v(o)}^{e_1}) \wedge \xi_o^{e_1}]   \vee 
  \\
       [\psi_{\neg x_t}\StrictUntil_{\geq 0}\{\psi(\Ending,v_t)  \wedge \EqTime(\psi(e_1,v(0))\wedge\xi_o^{e_1})\}]
\end{multline*}
 The first disjunct considers the case where the $e_1$-event of token $o$ occurs strictly between the start-event and the end-event of the trigger token $o_t$ (along the encoding of a multi-timeline of $SV$).
 The second considers the case where the  $e_1$-event of token $o$ precedes the start-event of the trigger token: thus, under the future semantics, it holds that
 the $e_1$-time of token $o$ coincides with the start time of the trigger token. Finally, the third disjunct considers the case where the $e_1$-event of token $o$ follows the
 end-event of the trigger token (in this case, the related timestamps have to coincide).
 \item $\rho = o_t \leq_I^{e_1,e_2} o$ and $o\neq o_t$. We assume $e_1=\Ending$  and $0\in I$ (the other cases being simpler). Then,
  \[
  \chi_\rho = [\psi_{\neg x_t}\StrictUntil_{\geq 0}(\psi(\Ending,u_t)\wedge \Eventually_I (\psi(e_2, v(o))\wedge \xi_o^{e_2}) ) ]\vee
  [\psi_{\neg x_t}\StrictUntil_{\geq 0}(\psi(\Ending,u_t)\wedge \Past_{v(o)}^{e_2} \wedge \xi_o^{e_2})],
  \]
  where the second disjunct captures the situation where the $e_2$-time  of $o$ coincides with the end time of the trigger token $o_t$, but the $e_2$-event of $o$ occurs before the end-event of the trigger token.
  \item $\rho = o_t \leq_I^{e_1,e_2} o_t$. This case is straightforward and we omit the details.
  \item $\rho = o_1 \leq_I^{e_1,e_2} o_2$, $o_1\neq o_t$ and $o_2 \neq o_t$. We assume $o_1\neq o_2$  and $0\in I$ (the other cases are simpler).  Then,
\begin{multline*}
  \chi_\rho \!:=\!  [\Past_{v(o_1)}^{e_1} \!\wedge \xi_o^{e_1} \!\!\wedge\!\! \Eventually_I (\psi(e_2,v(o_2)) \!\wedge\! \xi_o^{e_2}) ]  \! \vee\!    [\Eventually_{\geq 0}\{\psi(e_1,v(o_1)) \wedge \xi_o^{e_1} \wedge \Eventually_I (\psi(e_2,v(o_2)) \wedge \xi_o^{e_2})\} ]   \vee 
\\
    [\Past_{v(o_1)}^{e_1} \wedge \xi_o^{e_1} \wedge \Past_{v(o_2)}^{e_2} \wedge \xi_o^{e_2}] \vee [\Past_{v(o_2)}^{e_2} \wedge \xi_o^{e_2} \wedge \EqTime(\psi(e_1,v(o_1)) \wedge\xi_o^{e_1})]   \vee 
  \\
  [\Eventually_{\geq 0}\{\psi(e_2,v(o_2)) \wedge \xi_o^{e_2} \wedge \EqTime(\psi(e_1,v(o_1)) \wedge\xi_o^{e_1})\}]
\end{multline*}
The first two disjuncts handle the cases where (under the future semantics) the $e_1$-event of token $o_1$ precedes the $e_2$-event of token $o_2$, while the last three disjuncts consider the dual situation.
In the latter case, the $e_1$-time of token $o_1$ and the $e_2$-time of token $o_2$ are equal.
\end{compactitem}
Note that the \MTL\ formula $\varphi_\forall$ is an \MITL\ formula (resp., $\MITLR$ formula) if the intervals in the trigger rules are non-singular (resp., belong to $\IntvR$). This concludes the proof. 
 \end{proof}

\paragraph{Encoding of trigger-less rules by \TA.} 
We note that an existential statement in a trigger-less rule requires
the existence of an \emph{a priori bounded number} of temporal events  satisfying mutual temporal relations. Hence, one can easily construct a \TA\
which guesses such a chain of events and checks the temporal relations by clock constraints and clock resets.
Thus, by the well-known effective closure of \TA\ under language union and intersection~\cite{ALUR1994183}, we obtain the following result (for details, see~\cite{techrepGand}).

\newcounter{prop-TATriggerLessRules}
\setcounter{prop-TATriggerLessRules}{\value{proposition}}

 \begin{proposition}\label{prop:TATriggerLessRules} One can construct in exponential time a \TA\ $\Au_{_\exists}$ over $2^{\Prop}$  such that, for each multi-timeline $\Pi$ of $SV$ and encoding $w_\Pi$ of $\Pi$, $w_\Pi$ is accepted by $\Au_{\exists}$
  iff $\Pi$ satisfies all the  trigger-less  rules in $R$. Moreover, $\Au_{_\exists}$ has  $2^{O(N_q)}$ states, $O(N_q)$ clocks, and maximal constant $O(K_P)$, where
  $N_q$  is the overall number of quantifiers   in the trigger-less  rules of $R$.
 \end{proposition}

\paragraph{Conclusion of the construction.}
By applying Propositions~\ref{prop:AtutomataForMultiTimeline}--\ref{prop:TATriggerLessRules} and well-known results about \TA\ and \MTL\ over finite timed words~\cite{ALUR1994183,OuaknineW07},
we obtain the main result of this section.

\begin{theorem}\label{theorem:UpperBounds}
The future TP problem with simple trigger rules is decidable.
Moreover, if the intervals in the atoms of the trigger rules are non-singular
(resp., belong to $\Intv_{(0,\infty)}$), then the problem is in $\EXPSPACE$ (resp., in $\Psp$).
\end{theorem}
\begin{proof} 
We fix an instance $P=(SV,R)$ of the problem with maximal constant $K_P$.
Let $N_v := \sum_{x\in SV}|V_x|$, $N_q$ be the overall number of quantifiers in the trigger-less  rules of $R$, and $N_a$ be the overall number of atoms in the
trigger rules of $R$.
By Propositions~\ref{prop:AtutomataForMultiTimeline}--\ref{prop:TATriggerLessRules} and the effective closure
of \TA\ under language intersection~\cite{ALUR1994183}, we can build a \TA\ $\Au_P$  and an \MTL\ formula $\varphi_\forall$ such that there is a future plan of $P$ iff
$\TLang(\Au_P)\cap\TLang(\varphi_\forall)\neq \emptyset$. Moreover, $\Au_P$ has $2^{O(N_q+N_v)}$ states, $O(N_q+|SV|)$ clocks, and maximal constant $O(K_P)$, while $\varphi_\forall$ has $O(N_a+N_v)$ distinct subformulas and
maximal constant $O(K_P)$. By~\cite{OuaknineW07}, checking non-emptiness of  $\TLang(\Au_P)\cap\TLang(\varphi_\forall)$ is decidable. Hence, the first part of the theorem holds. For the
second part, assume that  the intervals in 
the trigger rules are non-singular
(resp., belong to $\Intv_{(0,\infty)}$). By Proposition~\ref{prop:MTLTriggerRules}, $\varphi_\forall$ is an \MITL\ (resp., $\MITLR$) formula. Thus, by~\cite{Alur:1996}, one can build a \TA\ $\Au_\forall$ accepting $\TLang(\varphi_\forall)$ having $ 2^{O(K_P\cdot(N_a+N_v))} $ states, $O(K_P\cdot(N_a+N_v))$ clocks
(resp., $O(2^{(N_a+N_v)})$ states, $O(N_a+N_v)$ clocks), and maximal constant $O(K_P)$.
Non-emptiness of a \TA\ $\Au$ can be solved by an $\NPSPACE$ search algorithm in the \emph{region graph} of $\Au$ which uses space logarithmic in the number of control 
states of $\Au$ and polynomial
in the number of clocks and in the length of the encoding of the maximal constant of $\Au$~\cite{ALUR1994183}.
  Thus, since $\Au_P$, $\Au_\forall$, and the intersection $\Au_\wedge$ of $\Au_P$ and $\Au_\forall$ can be constructed on the fly, and the search in the region graph of
$\Au_\wedge$ can be done without explicitly constructing $\Au_\wedge$, the result follows.
\end{proof} 

\section{Non-primitive recursive hardness of the future TP problem}\label{sec:NPRHardness}

In this section, we establish the following result.
\begin{theorem}\label{theorem:NPRHardness}
Future TP with \emph{one state variable} is non-primitive recursive-hard  even under one of the following two assumptions: \emph{either} (1) the trigger rules are simple,
\emph{or} (2) the intervals  are in $\Intv_{(0,\infty)}$.
\end{theorem}

Theorem~\ref{theorem:NPRHardness} is proved by a polynomial-time reduction from the halting problem for \emph{Gainy counter machines}~\cite{DemriL09}, a variant of standard Minsky machines, where  the counters may erroneously  increase. Fix such a machine $M = \tpl{Q,q_\init,q_\halt,n, \Delta}$,
where
\begin{inparaenum}[(i)]
  \item  $Q$ is a finite set of (control) locations, $q_\init\in Q$ is the initial location, and $q_\halt\in Q$ is the halting location,
  \item  $n\in\Nat\setminus\{0\}$ is the number of counters, and
  \item $\Delta \subseteq Q\times \Inst \times Q$ is a transition relation over the instruction set $\Inst= \{\inc,\dec,\zero\}\times \{1,\ldots,n\}$.
\end{inparaenum}
We adopt the following notational conventions.
 For an instruction
$\instr\in \Inst$, let $c(\instr)\in\{1,\ldots,n\}$ be the counter associated with $\instr$.
For a transition $\delta\in \Delta$ of the form $\delta=(q,\instr,q')$, define $\From(\delta):= q$, $\instr(\delta):=\instr$, $c(\delta):= c(\instr)$,
and $\To(\delta):= q'$.  We denote by $\instr_\init$ the instruction $(\zero,1)$.
W.l.o.g., we make these assumptions:
\begin{inparaenum}[(i)]
  \item for each transition $\delta\in \Delta$, $\From(\delta)\neq q_\halt$ and $\To(\delta )\neq q_\init$, and
    \item there is exactly one transition in $\Delta$, denoted $\delta_\init$, having as source the initial location $q_\init$.
\end{inparaenum}

An $M$-configuration is a pair $(q,\nu)$ consisting of a location $q\in Q$ and a counter valuation $\nu: \{1,\ldots,n\}\to \Nat$. Given
two valuations $\nu$ and $\nu'$, we write $\nu\geq \nu'$ iff $\nu(c)\geq \nu'(c)$ for all $c\in\{1,\ldots,n\}$.

\details{
Under the \emph{standard Minsky semantics} (with no errors), $M$ induces a transition relation, denoted by $\longrightarrow$, over pairs of $M$-configurations and instructions defined as follows:
 for configurations $(q,\nu)$ and $(q',\nu')$, and instructions $\instr \in \Inst$, $(q,\nu) \der{\instr} (q',\nu')$ if the following holds, where $c\in \{1,\ldots,n\}$ is the counter associated with the instruction
 $\instr$:
\begin{inparaenum}[(i)]
  \item  $(q,\instr,q')\in \Delta$ and $\nu'(c')= \nu(c')$ for all $c'\in \{1,\ldots,n\}\setminus\{c\}$;
  \item  $\nu'(c)= \nu(c) +1$ if $\instr=(\inc,c)$;
  \item $\nu'(c)= \nu(c) -1$ if $\instr=(\dec,c)$ (in particular, $v(c)>0$);
   \item  $\nu'(c)= \nu(c)=0$ if $\instr=(\zero,c)$.
\end{inparaenum}
}

The \emph{gainy semantics} is obtained from the standard Minsky semantics by allowing \emph{incrementing} errors.
Formally, 
$M$  induces a transition relation, denoted by $\longrightarrow_\gainy$, defined as follows:
for configurations $(q,\nu)$ and $(q',\nu')$, and instructions $\instr \in \Inst$, $(q,\nu) \derG{\instr} (q',\nu')$ if the following holds, where $c=c(\instr)$:
\details{
$(q,\nu) \derG{\instr} (q',\nu')$ iff there are valuations $\nu_+$ and $\nu'_+$ such that
$\nu_+ \geq \nu$, $(q,\nu_+) \der{\instr} (q',\nu'_+)$, and $\nu' \geq \nu'_+$. Equivalently, $(q,\nu) \derG{\instr} (q',\nu')$ iff
the following holds, where  $c\in \{1,\ldots,n\}$ is the counter associated with
 $\instr$:}
\begin{inparaenum}[(i)]
  \item  $(q,\instr,q')\in \Delta$ and $\nu'(c')\geq  \nu(c')$ for all $c'\in \{1,\ldots,n\}\setminus\{c\}$;
  \item  $\nu'(c)\geq  \nu(c) +1$ if $\instr=(\inc,c)$;
  \item $\nu'(c)\geq  \nu(c) -1$ if $\instr=(\dec,c)$;
   \item  $\nu(c)=0$ if $\instr=(\zero,c)$.
\end{inparaenum}

A (gainy) computation of $M$ is a finite sequence of global gainy transitions
\[
(q_0,\nu_0) \derG{\instr_0} (q_1,\nu_1) \derG{\instr_1} \cdots  \derG{\instr_{k-1}} (q_k,\nu_k)
\]
$M$ \emph{halts} if there is a computation starting at the \emph{initial} configuration $(q_\init, \nu_\init)$, where $\nu_\init(c) = 0$ for all $c\in \{1,\ldots,n\}$, and leading to some
halting configuration
$(q_{\halt}, \nu)$. The halting problem is to decide whether a given gainy machine $M$ halts, and it was proved to be decidable and non-primitive recursive~\cite{DemriL09}.
We prove the following result, from which Theorem~\ref{theorem:NPRHardness} directly follows.

\begin{proposition}\label{prop:NPRHardness}
One can construct in polynomial time a TP instance $P=(\{x_M\},R_M)$  s.t.\ the trigger rules in $R_M$ are simple (resp., the intervals in $P$ are in $\Intv_{(0,\infty)}$)
and $M$ halts iff there is a future plan of $P$.
\end{proposition}
\begin{proof}
%
We focus on the reduction where the intervals in $P$ are in $\Intv_{(0,\infty)}$. At the end of the proof, we show how to adapt the construction for 
the case of simple trigger rules with arbitrary intervals.

First, we define a suitable encoding of a computation of $M$ as a timeline for $x_M$. For this,
we exploit the finite set of symbols $V:= V_{\main}\cup V_{\cont}\cup V_{\dummy}$ corresponding to the finite domain of the state variable $x_M$.
The   sets of \emph{main} values $V_{\main}$ is given by
$
V_\main := \{(\delta,\instr)\in\Delta\times\Inst\mid \instr\neq (\inc,c) \text{ if }\instr(\delta)=(\zero,c)\}.
$
The set of \emph{secondary} values $V_{\cont}$ is defined as ($\#_\inc$ and $\#_\dec$
are two special symbols used as markers):
$
 V_\cont := V_\main \times \{1,\ldots,n\} \times 2^{\{\#_\inc,\#_\dec\}}.
$
Finally, the set of \emph{dummy values} is $(V_{\main}\cup V_{\cont})\times \{\dummy\}$.

Intuitively, in the encoding of an $M$-computation a main value $(\delta,\instr)$ keeps track of the transition $\delta$ used in the current step of the computation, while
$\instr$ represents the instruction exploited in the previous step (if any) of the computation.
The set $V_{\cont}$ is used for encoding counter values, while the set $V_{\dummy}$ is used for specifying punctual time constraints by means of
non-simple trigger rules over 
$\Intv_{(0,\infty)}$. For a word $w\in V^{*}$, we denote by $||w||$ the length of the word
obtained from $w$ by removing dummy symbols.

 For $c\in \{1,\ldots,n\}$ and $v_\main = (\delta,\instr)\in V_\main$,
the \emph{set $\Tag(c,v_\main)$ of markers of counter $c$ for the main value $v_\main$} is the subset of  $\{\#_\inc,\#_\dec\}$ defined as follows:
\begin{inparaenum}[(i)]
  \item $\#_\inc\in \Tag(c,v_\main)$ iff $\instr = (\inc,c)$;
  \item $\#_\dec\in \Tag(c,v_\main)$ iff $\instr(\delta) = (\dec,c)$;
\end{inparaenum}

A \emph{$c$-code for the main value $v_\main= (\delta,\instr)$} is a  finite word $w_c$ over $V_\cont$ such that
\emph{either} (i) $w_c$ is empty and $\#_\inc\notin\Tag(c,v_\main)$, \emph{or} (ii) $\instr(\delta)\neq (\zero,c)$ and $w_c=(v_\main,c,\Tag(c,v_\main))(v_\main,c,\emptyset,\dummy)^{h_0}\cdot (v_\main,c,\emptyset)\cdot (v_\main,c,\emptyset,\dummy)^{h_1}\cdots (v_\main,c,\emptyset)\cdot (v_\main,c,\emptyset,\dummy)^{h_n}$ for some $n\geq 0$ and $h_0,h_1,\ldots,\allowbreak h_n\geq 0$. The $c$-code $w_c$ encodes the value for counter $c$
given by $||w_c||$. Intuitively, $w_c$ can be seen as an interleaving of secondary values with dummy ones, the latter being present only for technical aspects, but not encoding any counter value.

A \emph{configuration-code $w$  for a main value $v_\main=(\delta,\instr)\in V_\main$} is a finite word over $V$
of the form $w= v_\main \cdot (v_\main,\dummy)^{h}\cdot w_1\ldots w_n$, where $h\geq 0$ and for each counter $c\in \{1,\ldots,n\}$, $w_c$ is a $c$-code
for the main value $v_\main$. The configuration-code $w$ encodes the $M$-configuration$(\From(\delta),\nu)$, where $\nu(c)=||w_c||$
for all $c\in \{1,\ldots,n\}$. Note that if $\instr(\delta)=(\zero,c)$, then $\nu(c)=0$ and $\instr\neq (\inc,c)$. Moreover,  the marker $\#_\inc$ occurs in $w$ iff $\instr$ is an increment instruction, and in such a case
 $\#_\inc$ marks the first symbol of the encoding $w_{c(\instr)}$ of counter $c(\instr)$. Intuitively, if the operation performed in the previous step
 of the computation increments counter $c$, then the tag $\#_\inc$ ``marks" the unit of the counter $c$ in the current configuration which has been added by the increment.
Additionally, the marker $\#_\dec$ occurs in $w$ iff $\delta$ is a decrement instruction and the value of counter $c(\delta)$ in $w$ is non-null; in such a case,
 $\#_\dec$ marks the first symbol of the encoding $w_{c(\delta)}$ of counter $c(\delta)$. Intuitively, if the operation to be performed in the current step
 decrements counter $c$ and the current value of 
 $c$ is non-null, then the tag $\#_\dec$  marks  the unit of the counter $c$  in the current configuration which has to be removed by the decrement.

A \emph{computation}-code is a sequence of configuration-codes $\pi= w_{(\delta_0 ,\instr_0)} \cdots w_{(\delta_k,\instr_k)}$, where,  for all $0\leq i\leq k$, $w_{(\delta_i,\instr_i)}$ is a configuration-code with main value $(\delta_i,\instr_i)$, and whenever
  $i<k$, it holds that $\To(\delta_i)=\From(\delta_{i+1})$ and $\instr(\delta_i)=\instr_{i+1}$. Note that by our assumptions $\To(\delta_i)\neq q_\halt$ for all $0\leq i<k$, and
  $\delta_j\neq \delta_\init$ for all $0<j\leq k$.
  The computation-code $\pi$ is \emph{initial} if  the first configuration-code $w_{(\delta_0 ,\instr_0)}$ is $(\delta_\init,\instr_\init)$ (which encodes the initial configuration), and it is \emph{halting} if
  for the last  configuration-code $w_{(\delta_k,\instr_k)}$ in $\pi$, it holds that $\To(\delta_k)=q_\halt$.
For all $0\leq i\leq k$, let $(q_i,\nu_i)$ be the $M$-configuration encoded by the configuration-code $w_{(\delta_i,\instr_i)}$ and $c_i= c(\delta_i)$.
 The computation-code $\pi$ is \emph{well-formed} if, additionally,   for all $0\leq j\leq k-1$,  the following holds:
\begin{inparaenum}[(i)]
  \item $\nu_{j+1}(c)\geq \nu_j(c)$ for all $c\in \{1,\ldots,n\}\setminus \{c_j\}$ (\emph{gainy monotonicity});
\item $\nu_{j+1}(c_j)\geq \nu_j(c_j)+1$ if $\instr(\delta_j)= (\inc,c_j)$ (\emph{increment req.});
 \item $\nu_{j+1}(c_j)\geq \nu_j(c_j)-1$ if $\instr(\delta_j)= (\dec,c_j)$ (\emph{decrement req.}).
\end{inparaenum}
Clearly, 
$M$ halts \emph{iff} there is an initial and halting well-formed computation-code. 

\paragraph*{Definition of $x_M$ and $R_M$.} We now define a state variable $x_M$ and a set $R_M$ of  synchronization rules for $x_M$ with intervals in $\Intv_{(0,\infty)}$ such that the untimed part of every \emph{future plan} of $P=(\{x_M\},R_M)$
is an initial and halting well-formed computation-code. Thus, $M$ halts iff there is a future plan of $P$.

Formally, 
variable $x_M$ is given by $x_M= (V,T,D)$, where, for each $v\in V$,
$D(v)=]0,\infty[$ if $v\notin V_{\dummy}$, and $D(v)=[0,\infty[$ otherwise. Thus, we require that the duration of a non-dummy token 
is always greater than zero (\emph{strict time monotonicity}).
The value transition function $T$ of $x_M$ ensures the following. 
\begin{claim}\label{ref:claim}
The untimed parts of the timelines for $x_M$ whose first token has value $(\delta_\init,\instr_\init)$ correspond
 to the prefixes of  initial computation-codes. Moreover, $(\delta_\init,\instr_\init)\notin T(v)$ for all $v\in V$.
\end{claim}

%

 By construction, it is a trivial task to define $T$ so that the previous requirement is fulfilled. Let $V_\halt=\{(\delta,\instr)\in V_\main\mid \To(\delta)=q_\halt\}$.
 By Claim~\ref{ref:claim} and the assumption that  $\From(\delta)\neq q_\halt$ for each transition $\delta\in \Delta$, for  the initialization and halting requirements,
 it suffices  to ensure that a timeline has a token with value $(\delta_\init,\instr_\init)$ and a token with value in $V_\halt$. This is captured by the trigger-less rules
   $
   \true \rightarrow \exists   o[x_M=(\delta_\init,\instr_\init)].  \true
   $ and  $\true \rightarrow \bigvee_{v\in V_\halt} \exists   o[x_M=v].  \true $.

Finally, the crucial well-formedness requirement is captured by the trigger rules in $R_M$ which express the following punctual time constraints.
Note that we take advantage of the dense temporal domain to allow
for the encoding of arbitrarily large values of counters in two time units.

%
 \begin{compactitem}
   %
   %
   %
   %
    %
   \item \emph{2-Time distance between consecutive main values:} the overall duration of the sequence of tokens corresponding to a configuration-code  amounts exactly to two time units. 
By Claim~\ref{ref:claim}, strict time monotonicity, and the halting requirement, it suffices to ensure that each token $tk$ having a  main value in $V_\main \setminus V_\halt$ is eventually followed by a token $tk'$  such that $tk'$ has a  main value and $\startTime(tk')-\startTime(tk)=2$. To this aim, for each  $v\in V_\main \setminus V_\halt$,  we have the following non-simple trigger rule with intervals in $\Intv_{(0,\infty)}$ which uses a dummy-token for capturing the punctual time constraint:
\[o[x_M=v] \rightarrow \bigvee_{u\in V_\main}\bigvee_{u_d\in V_\dummy} \exists  o'[x_M= u]\exists  o_d[x_M= u_d].  o\leq^{\start,\start}_{[1,+\infty[} o_d \,\wedge\, o_d\leq^{\start,\start}_{[1,+\infty[}o'\, \wedge\, o\leq^{\start,\start}_{[0,2]} o'.\]
   \item 
   For a counter $c\in \{1,\ldots,n\}$, let $V_c\subseteq V_\cont$ be the set of secondary states given
    by $V_\main \times \{c\} \times 2^{\{\#_\inc,\#_\dec\}}$. We require that  each token $tk$  with
    a $V_{c}$-value of the form $((\delta,\instr),c,\Tag)$  such that $c\neq c(\delta)$ and $\To(\delta)\neq q_\halt$ is eventually followed by a token $tk'$ with a $V_{c}$-value such that  $\startTime(tk')-\startTime(tk)=2$.
Note that our encoding, Claim~\ref{ref:claim}, strict time monotonicity, and 2-Time distance between consecutive main values guarantee  that the previous requirement captures \emph{gainy monotonicity}.
 Thus, for each counter $c$ and $v\in V_{c}$ such that $v$ is of the form $((\delta,\instr),c,\Tag)$, where $c\neq c(\delta)$ and $\To(\delta)\neq q_\halt$, we have the following non-simple trigger rule over $\Intv_{(0,\infty)}$:\\
\[o[x_M=v] \rightarrow \bigvee_{u\in V_{c}}\bigvee_{u_d\in V_\dummy} \exists  o'[x_M= u]\exists  o_d[x_M= u_d].  o\leq^{\start,\start}_{[1,+\infty[} o_d \,\wedge\, o_d\leq^{\start,\start}_{[1,+\infty[}o'\, \wedge\, o\leq^{\start,\start}_{[0,2]} o'.\]
  \item For capturing the increment and decrement requirements, by construction, it suffices to enforce that (i) each token $tk$  with
    a $V_{c}$-value of the form $((\delta,\instr),c,\Tag)$  such that $\To(\delta)\neq q_\halt$ and $\delta=(\inc,c)$ is eventually followed by a token $tk'$ with a $V_{c}$-value which is \emph{not} marked by the tag $\#_\inc$ such that  $\startTime(tk')-\startTime(tk)=2$, and (ii)
  each token $tk$  with
    a $V_{c}$-value of the form $((\delta,\instr),c,\Tag)$  such that $\To(\delta)\neq q_\halt$, $\delta=(\dec,c)$, and  $\#_\dec\notin \Tag$ is eventually followed by a token $tk'$ with a $V_{c}$-value  such that  $\startTime(tk')-\startTime(tk)=2$. These requirements can be expressed by non-simple trigger rules with intervals in $\Intv_{(0,\infty)}$ similar to the previous ones.
\end{compactitem}
Finally, to prove Proposition~\ref{prop:NPRHardness} for the case of simple trigger rules with arbitrary intervals, it suffices to remove the dummy values and replace the conjunction
  $o\leq^{\start,\start}_{[1,+\infty[} o_d \,\wedge\, o_d\leq^{\start,\start}_{[1,+\infty[}o'\, \wedge\, o\leq^{\start,\start}_{[0,2]} o'$ in the previous trigger
  rules with the \emph{punctual} atom $ o\leq^{\start,\start}_{[2,2]} o'$.
  This concludes the proof of Proposition~\ref{prop:NPRHardness}. \qedhere
\end{proof}

\section{Hardness of future TP with simple rules and non-singular intervals\label{sec:pspace}}

In this section, we first consider the future TP problem with simple trigger rules and non-singular intervals, and prove that it is $\EXPSPACE$-hard by a polynomial-time reduction from a \emph{domino-tiling problem for grids with rows of single exponential length}, which is known to be  $\EXPSPACE$-complete~\cite{harel92}. Since the reduction is standard, we refer the reader to~\cite{techrepGand} for the details of the construction.

\begin{theorem}
The future TP problem with simple trigger rules and non-singular intervals is $\EXPSPACE$-hard (under polynomial-time reductions).
\end{theorem}

We now focus on the special case 
with intervals of the forms $[0,a]$, with $a\in\Nat\setminus\{0\}$, and $[b,+\infty[$, with $b\in\Nat$, only, proving that it is $\Psp$-hard by reducing periodic SAT to it in polynomial time.

The problem \emph{periodic SAT} is defined as follows~\cite{1994-papadimitriou}.
We are given a Boolean formula $\varphi$ in \emph{conjunctive normal form},
defined over two sets of variables, $\Gamma=\{x_1,\ldots, x_n\}$ and $\Gamma^{+1}=\{x_1^{+1},\ldots , x_n^{+1}\}$, namely,
$
    \varphi = \bigwedge_{t=1}^m (\bigvee_{x\in (\Gamma \cup \Gamma^{+1})\cap L^+_t} x  \vee \bigvee_{x\in (\Gamma \cup \Gamma^{+1})\cap L^-_t}  \neg x),
$
where $m$ is the number of conjuncts of $\varphi$ and, for $1\leq t\leq m$, $L^+_t$ (resp., $L^-_t$) is the set of variables occurring non-negated (resp., negated) in the $t$-th conjunct of $\varphi$.
Moreover, the formula $\varphi^j$, for $j\in\Nat\setminus\{0\}$, is defined as $\varphi$ in which we replace each variable
$x_i\in \Gamma$ by a fresh one $x_i^j$, and $x_i^{+1}\in \Gamma^{+1}$ by $x_i^{j+1}$.
Periodic SAT is to decide the satisfiability of the (infinite-length) formula 
$\Phi= \bigwedge_{j\in\Nat\setminus\{0\}} \varphi^j$, that is, deciding the existence 
of a truth assignment of (infinitely many) variables $x_i^j$, for $i=1,\ldots, n,\, j\in\Nat\setminus\{0\}$, satisfying $\Phi$.
Periodic SAT is $\Psp$-complete~\cite{1994-papadimitriou}; in particular membership to such a class is proved by showing that one can equivalently check satisfiability of the (finite-length) formula $\Phi_f= \bigwedge_{j= 1}^{2^{2n}+1} \varphi^j$. Intuitively, $2^{2n}$ is the number of possible truth assignments to variables of $\Gamma\cup \Gamma^{+1}$, thus, after $2^{2n}+1$ copies of $\varphi$, we can find a repeated assignment: from that point, we can just copy the previous assignments. 
We now 
reduce periodic SAT to our problem.
Hardness also holds when only a single state variable is involved, and also restricting to intervals of the form $[0,a]$.

\begin{theorem}\label{theorem:PSPlowerBound} 
The future TP problem with simple trigger rules and intervals $[0,a]$, with $a\in\Nat\setminus\{0\}$, is $\Psp$-hard  (under polynomial-time reductions).
\end{theorem}
\begin{proof}
Let us define the state variable $y=(V,T,D)$, where 
\begin{compactenum}
    \item $V=\{\$,\tilde{\$},stop\}\cup \{x_i^\top,x_i^\bot, \tilde{x_i}^\top, \tilde{x_i}^\bot \mid i=1,\ldots, n\}$,
    \item $T(\$)=\{x_1^\top,x_1^\bot\}$, $T(\tilde{\$})=\{\tilde{x_1}^\top, \tilde{x_1}^\bot\}$ and $T(stop)=\{stop\}$,
    \item for $i=1,\ldots, n-1$, $T(x_i^\top)=T(x_i^\bot)=\{x_{i+1}^\top,x_{i+1}^\bot\}$,
    \item for $i=1,\ldots, n-1$, $T(\tilde{x_i}^\top)=T(\tilde{x_i}^\bot)=\{\tilde{x_{i+1}}^\top,\tilde{x_{i+1}}^\bot\}$,
    \item $T(x_n^\top)=T(x_n^\bot)=\{\tilde{\$},stop\}$,
    \item $T(\tilde{x_n}^\top)=T(\tilde{x_n}^\bot)=\{\$,stop\}$, and
    \item for all $v\in\ V$, $D(v)=[2,+\infty[$.
\end{compactenum}
Intuitively, we represent an assignment of variables $x_i^j$ by means of a timeline for $y$:
after every occurrence of the symbol $\$$, $n$ tokens are present, one for each $x_i$, and the value $x_i^\top$ (resp., $x_i^\bot$) represents a positive (resp., negative) assignment of $x_i^j$, for some \emph{odd} $j\geq 1$. Then, there is an occurrence of $\tilde{\$}$, after which $n$ more tokens occur, again one for each $x_i$, and the value $\tilde{x_i}^\top$ (resp., $\tilde{x_i}^\bot$) represents a positive (resp., negative) assignment of $x_i^j$, for some \emph{even} $j\geq 2$.
See Figure~\ref{fig:phij} for an example.
\begin{figure}[t]
    \centering
    \includegraphics[scale=0.7]{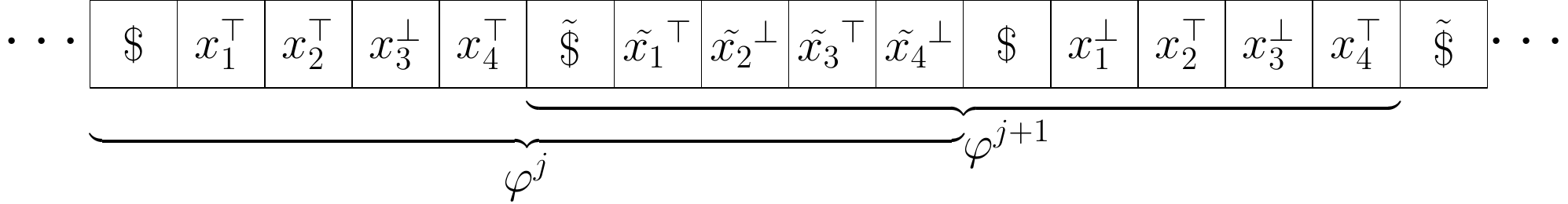}
    \vspace{-0.4cm}
    \caption{Let
    the formula $\varphi$ be defined over two sets of variables, $\Gamma=\{x_1,x_2,x_3,x_4\}$ and $\Gamma^{+1}=\{x_1^{+1},x_2^{+1},x_3^{+1},x_4^{+1}\}$. 
    The $j$-th copy (we assume $j$ is odd) of $\varphi$, i.e., $\varphi^j$, is satisfied by the assignment $x_1^j\mapsto \top$, $x_2^j\mapsto \top$, $x_3^j\mapsto \bot$, $x_4^j\mapsto \top$, $x_1^{j+1}\mapsto \top$, $x_2^{j+1}\mapsto \bot$, $x_3^{j+1}\mapsto \top$, $x_4^{j+1}\mapsto \bot$. The analogous for $\varphi^{j+1}$. }
    \label{fig:phij}
\end{figure}

We start with the next simple trigger rules, one for each $v\in V$: $o[y=v]\to o\leq^{\mathsf{s},\mathsf{e}}_{[0,2]} o$. Paired with the function $D$, they enforce
all tokens' durations to be \emph{exactly} 2: intuitively, since we exclude singular intervals, requiring, for instance, that a token $o'$ starts  $t$ instants of time after the end of $o$, with $t\in [\ell,\ell+1]$ and $\ell\in\Nat$ is even, boils down to $o'$ starting \emph{exactly} $\ell$ instants after the end of $o$. We also observe that, given the constant token duration, the density of the time domain does not play any role in this proof.

We now add the rules:
\begin{inparaenum}[(i)]
\item
    $\true \to \exists o[y=\$]. o\geq^{\mathsf{s}}_{[0,1]} 0$;
\item
    $\true \to \exists o[y=\tilde{\$}]. o\geq^{\mathsf{s}}_{[0,1]} (2^{2n}+1)\cdot 2(n+1)$;
\item
    $\true \to \exists o[y=stop]. o\geq^{\mathsf{s}}_{[0,1]} (2^{2n}+2)\cdot 2(n+1)$.
\end{inparaenum}
They respectively impose that $(i)$~a token with value $\$$ starts exactly at $t=0$ (recall that the duration of every token is 2);
$(ii)$~there exists a token with value $\tilde{\$}$ starting at $t=(2^{2n}+1)\cdot 2(n+1)$; 
$(iii)$~a token with value $stop$ starts at $t=(2^{2n}+2)\cdot 2(n+1)$. 
We are forcing the timeline to encode truth assignments for variables $x_1^1,\ldots, x_n^1,\ldots, x_1^{2^{2n}+2} ,\ldots , x_n^{2^{2n}+2}$: as a matter of fact, we will decide satisfiability of the finite formula $\Phi_f= \bigwedge_{j= 1}^{2^{2n}+1} \varphi^j$, which is equivalent to $\Phi$.


We now consider the next rules, that enforce the satisfaction of each $\varphi^j$ or,
equivalently, of $\varphi$ over the assignments of $(x_1^j,\ldots, x_n^j, x_1^{j+1},\ldots, x_n^{j+1})$.
For the $t$-th conjunct of $\varphi$, 
we define the future simple rule:
\begin{multline*}
    o[y=\tilde{\$}]\to \Big(\smashoperator{\bigvee_{x_i\in \Gamma\cap L^+_t}} \exists o'[y=\tilde{x_i}^\top]. o\leq ^{\mathsf{e},\mathsf{s}}_{[0,4n]} o' \Big) \vee 
        \Big(\smashoperator{\bigvee_{x_i^{+1}\in \Gamma^{+1}\cap L^+_t}} \exists o'[y=x_i^\top]. o\leq ^{\mathsf{e},\mathsf{s}}_{[0,4n]} o' \Big) \vee \\
        \Big(\smashoperator{\bigvee_{x_i\in \Gamma\cap L^-_t}} \exists o'[y=\tilde{x_i}^\bot]. o\leq ^{\mathsf{e},\mathsf{s}}_{[0,4n]} o' \Big) \vee 
        \Big(\smashoperator{\bigvee_{x_i^{+1}\in \Gamma^{+1}\cap L^-_t}} \exists o'[y=x_i^\bot]. o\leq ^{\mathsf{e},\mathsf{s}}_{[0,4n]} o' \Big) \vee \\
        \exists o''[y=stop]. o\leq ^{\mathsf{e},\mathsf{s}}_{[0,2n]} o''.
\end{multline*}
Basically, this rule (the rule where the trigger has value $\$$ being analogous) states that, after every occurrence of $\tilde{\$}$, a token $o'$, making true at least a (positive or negative) literal in the conjunct, must occur by $4n$ time instants (i.e., before the following occurrence of $\tilde{\$}$).
The disjunct $\exists o''[y=stop]. o\leq ^{\mathsf{e},\mathsf{s}}_{[0,2n]} o'' $ is present just to avoid evaluating $\varphi$ on the $n$ tokens before (the first occurrence of) $stop$.

The variable $y$ and all synchronization rules can be generated in time polynomial in $|\varphi|$ (in particular, all interval bounds and time constants of time-point atoms have a value, encoded in binary, in $O(2^{2n})$).
\end{proof}
\section{Conclusions and future work}
In this paper, we investigated decidability and complexity issues for TP over dense temporal domains. Such a problem is known to be undecidable~\cite{kr18} even if restricted to simple trigger rules. 
Here, we have shown that decidability can be recovered by adding the future semantics to simple trigger rules.
Moreover, future TP with simple trigger rules has been proved to be non-primitive recursive-hard 
(the same result holds in the case of future TP with all intervals being in $\Intv_{(0,\infty)}$).
Finally, if, additionally, singular intervals are avoided, it turns out to be $\EXPSPACE$-complete, 
and $\Psp$-complete if we consider only intervals in $\Intv_{(0,\infty)}$.

Future work will focus on decidability of future TP with arbitrary trigger rules which remains open.

\bibliographystyle{eptcs}
\bibliography{bib}
\end{document}